\documentclass[letterpaper]{article}
\usepackage{aaai}
\nocopyright
\usepackage{times}
\usepackage{helvet}
\usepackage{courier}
\newcommand{\ourtitle}{Risk Dynamics in Trade Networks}
\usepackage[draft,multiuser]{fixme} 
\fxusetheme{colorsig}
\FXRegisterAuthor{mr}{amr}{MR}	
\FXRegisterAuthor{rf}{arf}{RF}	
\newcommand{\raf}[1]{\begin{arfnote}{}#1\end{arfnote}}
\newcommand{\mdr}[1]{\begin{amrnote}{}#1\end{amrnote}}
\let\oldrfnote\rfnote
\renewcommand{\rfnote}[1]{\oldrfnote{\scriptsize #1}}
\renewcommand{\rfnote}[1]{} \renewcommand{\raf}[1]{} \renewcommand{\mdr}[1]{}

\usepackage{amsthm}
\usepackage{amsmath}
\usepackage{amssymb}
\usepackage{dsfont}
\usepackage{booktabs}
\usepackage{natbib}
\usepackage[pdftex]{graphicx}
\usepackage{url, verbatim}
\usepackage{tikz}
\usetikzlibrary{arrows,shapes,positioning}
\usetikzlibrary{shapes.multipart}

\usepackage{algorithm, algorithmicx}
\usepackage{algpseudocode}
\newcommand*\Let[2]{\State #1 $\gets$ #2}

\newcommand{\R}{\mathbb{R}}
\let\reals\R
\newcommand{\ones}{\mathds{1}}
\newcommand{\infc}{\wedge}

\newcommand{\tr}{\top}

\newcommand{\outcomes}{\Omega}
\newcommand{\dists}{\Delta}
\newcommand{\risk}{\rho}
\newcommand{\cash}{\ensuremath{r^{\raisebox{0pt}{$_\$$}}}}

\newcommand{\price}{\textrm{price}}
\newcommand{\positions}{\Xc}
\newcommand{\pos}{\mathcal{R}}
\newcommand{\net}{\Gc}
\newcommand{\nodes}{V}
\newcommand{\edges}{E}

\newcommand{\oo}{\omega}

\newcommand{\im}{\mathrm{im}}
\newcommand{\conv}{\mathrm{conv}}
\newcommand{\relint}{\mathrm{relint}}

\DeclareMathOperator*{\argmin}{\arg\min}
\DeclareMathOperator*{\argmax}{\arg\max}

\newcommand{\inner}[1]{\left\langle#1\right\rangle}

\newcommand{\E}[2]{\mathbb{E}_{#1}\left[#2\right]}

\newcommand{\Cc}{\mathcal{C}}

\newcommand{\Gc}{\mathcal{G}}
\newcommand{\Lc}{\mathcal{L}}

\newcommand{\Rc}{\mathcal{R}}
\newcommand{\Sc}{\mathcal{S}}

\newcommand{\Xc}{\mathcal{X}}

\newcommand{\eg}{\textit{e.g.}}

\newcommand{\ie}{\textit{i.e.}}
\newcommand{\aip}{{A_i}^{\!\!+}}

\newtheorem{definition}{Definition}

\newtheorem{theorem}{Theorem}
\newtheorem{lemma}{Lemma}

\begin{document}
\title{\ourtitle}
\author{Rafael M. Frongillo\\Harvard University\\\url{raf@cs.berkeley.edu} \And Mark D. Reid\\The Australian National University \& NICTA \\\url{mark.reid@anu.edu.au}}
\maketitle
\begin{abstract}
  We introduce a new framework to model interactions among agents which seek to trade to minimize their risk with respect to some future outcome.  We quantify this risk using the concept of risk measures from finance, and introduce a class of trade dynamics which allow agents to trade contracts contingent upon the future outcome.  We then show that these trade dynamics exactly correspond to a variant of randomized coordinate descent.  By extending the analysis of these coordinate descent methods to account for our more organic setting, we are able to show convergence rates for very general trade dynamics, showing that the market or network converges to a unique steady state.  Applying these results to prediction markets, we expand on recent results by adding convergence rates and general aggregation properties.  Finally, we illustrate the generality of our framework by applying it to agent interactions on a scale-free network.
\end{abstract}

\section{Introduction} 

The study of dynamic interactions between agents who each have a different stake in the world is of broad interest, especially in areas such as multiagent systems, decision theory, and economics.  In this paper, we present a new way to model such dynamic interactions, based on the notion of risk measures from the finance literature.

The agents in our model will each hold a \emph{position}, which states how much the agent stands to gain or lose for each possible outcome of the world.  The overall outlook of an agent's position will be quantified by their \emph{risk measure}, which simply captures the ``riskiness'' of their position.  To minimize their risks, agents change their positions by trading \emph{contingent contracts} amongst themselves; these contracts state that the owner is entitled to some amount of money which depends on this future outcome.  Traders can be thought of as residing on a \emph{network}, the edges or hyperedges of which dictate which agents can trade directly.

This simple setting gives rise to several natural questions, which we would like to understand:  Given a set of agents with initial positions, can a stable equilibrium be found, where no agents can trade further for mutual benefit?  If such an equilibrium exists, can the agents converge to it using a trading protocol, and if so what is the rate of convergence?  How does the structure of the underlying network affect change these answers?  What is the distribution of the agents' risks at equilibrium, and how does an agent's final risk depend on his position in the network?  This paper addresses and provides answers to many of these questions.

Our model is heavily inspired by the work of~\cite{Hu:2014}, who use risk-measure agents to draw connections between machine learning and prediction markets.  Another motivation comes from~\cite{Abernethy:2014}, who study a prediction market setting with risk-averse traders whose beliefs over the outcomes are members of an exponential family of distributions.  Both papers analyze the steady-state equilibrium of the market, leaving open the question of how, and how fast, the market may arrive at that equilibrium.  In fact, both papers specifically point to rates and conditions for convergence in their future work.

The contributions of this paper are threefold.  First, we develop a natural framework to model the interactions of networked agents with outcome-contingent utilities, which is tractable enough to answer many of the questions posed above.  Second, by showing that our trading dynamics can be recast as a randomized coordinate descent algorithm, we establish convergence rates for trading networks and/or agent models which are more general than the two prediction market papers above.  Third, along the way to showing our rates, we adapt and generalize existing coordinate descent algorithms from the optimization literature, \eg~\cite{nesterov2012efficiency} and \cite{richtarik2014iteration}, which may be of independent interest.

\raf{Extra bibliographic notes:

  Risk measures introduced by~\cite{artzner1999coherent}

  Risk measures first connected to prediction markets by (?) P. Carr, H. Geman, and D. Madan. Pricing and hedging in incomplete markets. Journal
of Financial Economics, 62(1):131–167, 2001 -- check~\cite{othman2011liquidity-sensitive}.

  Certainty equivalents are not always convex or concave~\cite{meucci2009risk}.}



\section{Setting}
\label{sec:setting}

Let $\outcomes$ be a finite set of possible outcomes.  Following \citep{Follmer:2004}, a \emph{position} is simply a function from outcomes to the reals, $X : \outcomes \to \R$.  Positions can be thought of as random variables which are intended to represent outcome-contingent 
monetary values.
 Denote by $\ones : \outcomes\to\reals$ the constant position with $\ones:\oo\mapsto 1$.
The set of all positions under consideration will be denoted $\positions$
and will be assumed to be closed under linear combination and contain at
least all the outcome-independent positions $\{\alpha\ones : \alpha\in\reals\}$.
We will denote by $\dists$ the set of probability distributions over $\outcomes$, namely $\dists = \{ p \in {[0,1]}^\outcomes \colon \inner{p, \ones} = 1 \}$, where $\inner{p,x} = \sum_{\oo\in\outcomes} p(\oo) x(\oo)$ is the inner product.
Note that $\inner{p, X} = \E{\oo\sim p}{X(\oo)}$, the mean under $p$.

When viewed as a vector space in $\reals^\outcomes$, the set of positions $\positions$ introduced above is a subspace of dimension $k\leq |\outcomes|$.  Hence, it must have a basis of size $k$, and thus we must have some $\phi:\outcomes\to\reals^k$ with the property that for all $X\in\positions$, there is some $r\in\reals^k$ such that $X(\oo) = r\cdot\phi(\oo) = \sum_i r_i\phi(\oo)_i$ for all $\oo\in\outcomes$.

We will make extensive use of this ``compressed'' form of $\positions$, which we denote by $\pos=\reals^k$.  Define the counterpart $X[r]\in\positions$ of $r\in\pos$ to be the position $X[r]:\oo\mapsto r\cdot\phi(\oo)$.  The presence of outcome-independent positions then translates into the existence of some $\cash \in\pos$ satisfying $X[\cash] = \ones$.  Finally, we denote by $\Pi := \conv(\phi(\outcomes))$ the convex hull of the basis function $\phi$.

As intuition about $\phi$ and $\Pi$, it is helpful to draw analogy to the setting of prediction markets.  As we will see in Section~\ref{sec:risknets-appl-pred-mark}, the function $\phi$ can be thought of as encoding the payoffs of each of $k$ outcome-contingent contracts, or \emph{securities}, where contract $i$ pays $\phi(\oo)_i$ for outcome $\oo$.  The space $\Pi$ then becomes the set of possible beliefs $\{\inner{p,\phi} : p\in\dists\}$ of the expected value of the securities.  

\subsection{Risk Measures} 

Following \citet{Hu:2014}, agents in our framework will each quantify their 
uncertainty in positions via a (convex monetary) \emph{risk measure} 
$\risk : \pos \to \R$ satisfying, for all $r, r' \in \positions$:
\begin{itemize}
	\item \emph{Monotonicity}: 
		$\forall\oo$ $X[r](\oo) \le X[r'](\oo)$ $\Rightarrow$
		$\risk(r) \ge \risk(r')$. 
	\item \emph{Cash invariance}: 
		$\risk(r + c\cdot\cash) = \risk(r) - c$ 
		for all $c \in \reals$.
	\item \emph{Convexity}: 
		$\risk(\lambda r + (1-\lambda) r') 
			\le \lambda \risk(r) + (1-\lambda) \risk(r')$\\
		for all $\lambda \in [0,1]$.
	\item \emph{Normalization}: $\risk(0) = 0$.
\end{itemize}
The reasonableness of these properties is usually argued as follows 
(see, \eg, \citep{Follmer:2004}). Monotonicity ensures that positions that
result in strictly smaller payoffs regardless of the outcome are considered
more risky. Cash invariance captures the idea that if a guaranteed payment
of $\$c$ is added to the payment on each outcome then the risk will decrease
by $\$c$. Convexity states that merging positions results in lower risk.   Finally, normalization is for convenience, stating that a
position with no payout should carry no risk.

In addition to these common assumptions, we will make two regularity assumptions:
\begin{itemize}
\item \emph{Expressiveness}: $\risk$ is everywhere differentiable, and $\mathrm{closure} \{\nabla \risk(r):r\in\pos\} = \Pi$.
\item \emph{Strict risk aversion}: the convexity inequality above is strict unless $r-r'= \lambda\cash$ for some $\lambda\in\reals$.
\end{itemize}
Expressiveness is related to the dual formulation given below; roughly, it says that the agent must take into account every possible distribution over outcomes when calculating the risk.  Strict risk aversion says that an agent should strictly prefer a mixture of positions, unless of course the difference is outcome-independent.

A key result concerning convex risk measures is the following 
representation theorem (cf. \citet[Theorem 4.15]{Follmer:2004}, \citet[Theorem 3.2]{abernethy2013efficient}).

\begin{theorem}[Convex Risk Representation]
  \label{thm:cvx-risk-dual}
	A functional $\rho : \pos \to \R$ is a convex risk measure if and 
	only if there is a closed convex function $\alpha : \Pi \to \R\cup\{\infty\}$ such that
	\begin{equation}\label{eq:convex-risk}
		\rho(r) = \sup_{\pi\in\relint(\Pi)} \inner{\pi,-r} - \alpha(\pi).
	\end{equation}
\end{theorem}
Here $\relint(\Pi)$ denotes the relative interior of $\Pi$, the interior relative to the affine hull of $\Pi$.
Notice that if $f^*$ denotes the convex conjugate
$f^*(y) := \sup_{x} \inner{y,x} - f(x)$, then this theorem states that
$\risk(r) = \alpha^*(-r)$.
This result suggests that the function $\alpha$ can be interpreted as a \emph{penalty function}, assigning
a measure of ``unlikeliness'' $\alpha(\pi)$ to each expected value $\pi$ of the securities defined above.  Equivalently, $\alpha(\inner{p,\phi})$ measures the unlikeliness of distribution $p$ over the outcomes.
We can then see that the risk is the greatest expected loss under each distribution, taking into account the penalties assigned by $\alpha$.

\raf{WAS HERE: examples}

\raf{Add a comment or two about the relationship to von Neumann - Morgenstern utility theory and risk aversion (forward ref to PM section)}


\subsection{Risk-Based Agents} 

We are interested in the interaction between two or more agents who express
their preferences for positions via risk measures.
\citet{Burgert:2006} formalise this problem by considering $N$ agents with
risk measures $\risk_i$ for $i = 1, \ldots, N$ and asking how to split a
position $r \in \pos$ in to per-agent positions $r_i$ satisfying 
$\sum_i r_i = r$ so as to minimise the total risk $\sum_i \risk_i(r_i)$. 
They note that the value of the total risk is given by the 
\emph{infimal convolution} $\infc_i \risk_i$ of the individual agent risks
--- that is, 
\begin{equation}
  \label{eq:risknets-inf-conv}
  (\infc_i \risk_i)(r) := \inf \left\{ \sum_i \risk_i(r_i) : \sum_i r_i = r \right\}.    
\end{equation}

A key property of the infimal convolution, which will underly much of our analysis, is that its convex conjugate is the sum of the conjugates of its constituent functions.  See \eg~\cite{Rockafellar:1997} for a proof.

\begin{equation}
  \label{eq:risknets-inf-conv-dual}
  (\infc_i \risk_i)^* = \sum_i \risk_i^*~.
\end{equation}

\rfnote{Added for arXiv version}
As a brief aside, we note that one could think of $\infc_i \risk_i$ as the ``market risk'', which captures the risk of the entire market as if it were one entity.  By definition, eq.~\eqref{eq:risknets-inf-conv} says that the market is trying to reallocate the risk so as to minimize the net risk.  This interpretation is confirmed by eq.~\eqref{eq:risknets-inf-conv-dual} when we interpret the duals as penalty functions as above: the penalty of $\pi$ is the sum of the penalties of the market participants.  This collective view is useful when thinking about hierarchical markets, as we will discuss briefly in Section~\ref{sec:risknets-appl-pred-mark}.

\citet{Hu:2014} identify a special, market making agent with risk $\risk_0$
that aims to keep its risk constant rather than minimising it.
The risk minimising agents trade with the market maker by paying the market
maker $\risk_0(-r)$ dollars in exchange for receiving position $r$, thus keeping the market maker's risk constant.  We will revisit these special constant-risk interactions in Section~\ref{sec:risknets-appl-pred-mark}.  For now, we will consider quite general trading dynamics.


\section{Trade Dynamics}

We now describe how agents may interact with one another, by introducing certain dynamics of trading among agents.  Recall that we have $N$ agents, and each agent $i$ is endowed with a risk measure $\risk_i$.  We further endow agent $i$ with an initial position $r_i^0\in\pos$, and let $r^0 = \sum_i r_i^0$.  We will start time at $t=0$ and denote the position of trader $i$ at time $t$ by $r^t_i$.

A crucial concept throughout the paper is that of surplus.  Given a subset of the agents willing to trade among themselves, we can quantify the total net drop in risk that group can achieve.
\begin{definition}
  \label{def:surplus}
  Given $r_S = \{r_i\}_{i\in S}$ for some subset of agents $S$, the \emph{$S$-surplus} of $r$ is the function $\Phi:\pos^{|S|}\to\reals$ defined by
$\Phi_S(r_S) = \sum_{i\in S} \risk_i(r_i) - (\wedge_i \risk_i)(\sum_{i\in S} r_i)$.  In particular, $\Phi(r) := \Phi_{[N]}(r)$ is the \emph{surplus function}.
\end{definition}

We now define trade functions, which are efficient in the sense that all of this surplus is divided, perhaps unevenly, among the agents present.  A trade dynamic will then be simply a distribution over trade functions.

\begin{definition}
  \label{def:risknets-trade-func}
  Given some subset of nodes $S\subseteq [N]$, we say a function
  $f:\pos^N\to\pos^N$ is a \emph{trade function on $S$} if
  \begin{enumerate}
  \item $\sum_{i\in S} f(r)_i = \sum_{i\in S} r_i$,
  \item the $S$-surplus is allocated, meaning $\Phi_S(f(r)_S)=0$,
  \item for all $j\notin S$ we have $f(r)_j = r_j$.
  \end{enumerate}
\end{definition}

The following result shows that trade functions have remarkable structure: once the subset $S$ is specified, the trade function is completely determined, up to cash transfers.  In other words, the surplus is removed from the position vectors, and then it is redistributed as cash to the traders, and the choice of trade function is merely in how this redistribution takes place.  The fact that the derivatives match has strong intuition from prediction markets: agents must agree on the price.

\begin{theorem}
  \label{thm:risknets-trade-function}
  The trade functions on any $S\subseteq[N]$ are unique up to zero-sum cash transfers.
  Moreover, if $f$ is a trade function on $S$, then $\nabla\risk_i(f(r)_i) = \pi_S^*$ for all $i$, where
$\pi_S^* = \min_{\pi\in\Pi} \sum_{i\in S} \alpha_i(\pi) - \inner{\pi,\sum_{i\in S}r_i}$.
\end{theorem}
\begin{proof}
  By \cite[eq. X1.3.4.5]{hiriart1993grundlehren}, which gives a very general result about infimal convolutions, we have that the condition $\Phi_S(f(r)_S) = 0$ implies the existence of some $\pi$ such that $\nabla \risk_i(f(r)_i) = \pi$ for all $i\in S$.  The fact that $\pi = \pi_S^*$ follows by eq.~\eqref{eq:risknets-inf-conv-dual} and the definition of the conjugate.  We can now characterize such trade functions: by cash invariance, it is clear that $\nabla \risk_i(f(r)_i) = \nabla \risk_i(f(r)_i+c_i\cash)$ for all $c_i\in\reals$, and the strict risk aversion property says that these are the only such positions with the same derivative (otherwise convexity would imply $\risk_i$ is flat in between, a contradiction).  The requirement that $\sum_i f(r)_i = \sum_i r_i$ ensures $\sum_i c_i = 0$, meaning $f(r)$ is unique up to zero-sum cash transfers.
\end{proof}

Our notion of trade dynamics, defined below, is quite intuitive --- predefined groups of agents $S_i$ gather at random to negotiate a trade which minimizes their total risk, subject to the constraint that trading may only be among those gathered. 

\begin{definition}
  \label{def:risknets-trade-dynamic}
  Given $m$ subsets $\Sc = \{S_i\}_{i=1}^m$ and $m$ trade functions $f_i$ on $S_i$, and a distribution $p\in\Delta_m$ with full support, a \emph{trade dynamic} is the randomized algorithm which selects $f_i$ with probability $p_i$ and takes $r^{t+1} = f_i(r^t)$.  A \emph{fixed point} $r$ of the trade dynamic is a point with $f_i(r) = r$ for all $i\in[m]$.
\end{definition}

We now give a few natural instantiations of trade dynamics which we will use throughout the paper.

Let $G$ be a directed graph with a node for each agent.  An \emph{edge dynamic} has a trade function $f_{(i,j)}$ on $\{i,j\}$ for each edge $(i,j)$ in $G$, where if $r'=f_{(i,j)}(r)$ we have $\risk_j(r_j') = \risk_j(r_j)$ and $\risk_i(r_i') = \risk_i(r_i) - \Phi_{\{i,j\}}(r_{\{i,j\}})$.  In other words, the agents minimize their collective risks, but agent $i$ takes all of the surplus.  Similarly, a \emph{node dynamic} has a trade function $f_i$ for each node $i\in[N]$, on $S_i = \{j : (i,j)\in E(G)\} \cup \{i\}$, the out-neighborhood of $i$, and $r'=f_i(r)$ satisfies $\risk_j(r_j') = \risk_j(r_j)$ for $j\in S_i \setminus \{i\}$ while $\risk_i(r_i') = \risk_i(r_i) - \Phi_{S_i}(r_{S_i})$.

A third dynamic we will consider uses a notion of fairness; call a trade function $f$ on $S$ \emph{fair} if it satisfies $\risk_i(f(r)_i) = \risk_i(r_i) - \tfrac 1 {|S|} \Phi_S(r_S)$ for all $i\in S$.  Then a \emph{fair trade dynamic} is simply a mixture of fair trade functions.  Returning to the graph theme, we may define fair versions of the node and edge dynamics above, in the natural way.

For all these types of trade dynamics, we will see that the only crucial property is that of \emph{connectedness}, which ensures that trades can eventually travel from one agent to any other.  Given this property, we show a quite general equilibrium result.

\begin{definition}
  \label{def:risknets-dynamic-connected}
  A trade dynamic with subsets $\Sc$ is \emph{connected} if the hypergraph with nodes $[N]$ and hyperedges $\Sc$ is a connected hypergraph.
\end{definition}

\begin{theorem}
  \label{thm:dynamics-equilibrium}
  Let $\pi^* = \min_{\pi\in\Pi} \sum_i \alpha_i(\pi) - \inner{\pi,r^0}$.  There exists $r^*\in\pos^N$ such that for all connected trade dynamics $D$, the unique fixed point of $D$ is $r^*$, up to zero-sum cash transfers. Moreover, $\Phi(r^*) = 0$ and $\nabla \risk_i(r_i^*) = \pi^*$ for all $i$.
\end{theorem}
\begin{proof}
  Let $D = (\{S_i\}_{i=1}^m,\{f_i\}_{i=1}^m,p)$ be any connected trade dynamic, and assume $f_i(r) = r$ for all $i$.  This implies $\Phi_{S_i}(r_{S_i}) = 0$ for all $i$.  By Theorem~\ref{thm:risknets-trade-function},
this means that for all $i$ and all $j,j'\in S_i$ we have $\nabla \risk_j(r_j) = \nabla \risk_{j'}(r_{j'})$.  This gives us $m$ equivalence classes of derivatives, and by connectedness of the trade dynamic, we in fact have some $\pi$ for which $\nabla \risk_j(r_j) = \pi$ for all $j\in[N]$.  We can now appeal to \cite[Prop X1.3.4.2]{hiriart1993grundlehren}, which again is for general infimal convolutions, to conclude $\Phi(r) = 0$.  As this was the starting point in the proof of Theorem~\ref{thm:risknets-trade-function}, we immediately obtain $\pi=\pi^*$, and $r=r^*$ up to cash transfers.

\end{proof}

The result of Theorem~\ref{thm:dynamics-equilibrium} is somewhat surprising --- not only is there a unique equilibrium (up to cash transfers) for all connected dynamics, but all connected dynamics have the \emph{same} equilibrium!  If one restricts to connected graphical networks, this means that the equilibrium does not depend on the network structure.  The power of our framework is that the equilibrium analysis holds regardless of the way agents interact, as long as information is allowed to spread to all agents eventually.  In fact, one could even consider an arbitrary process choosing subsets $S^t$ of agents to trade at each time $t$; if the set $\Sc$ of subsets which are visited infinitely often yields a connected hypergraph, then the proof Theorem~\ref{thm:dynamics-equilibrium} still applies.

\section{Rates via Coordinate Descent}

Now that the existence of an equilibrium has been established, we turn to the question of convergence.  In this section, we will show that our trade dynamics are performing a type of randomized coordinate descent algorithm, where the coordinate subspaces correspond to subsets $S$ of agents.  
Our coordinate subspaces are more general than is currently considered in the literature, so to derive our convergence rates we will first need to  introduce a generalization of existing coordinate descent methods.  Using standard techniques to analyze it, we will arrive at convergence rates for very general classes of our dynamics.

Before reviewing the literature on coordinate descent, let us briefly see why this is a useful way to think of our dynamics.  Recall that we have $m$ subsets of agents $S_i$, and that each trade function $f_i$ only modifies the positions of agents in $S_i$.  Thinking of $(r_1,\ldots,r_N)$ as a large $N k$ vector (recall $\pos = \reals^k$), the trade function $f_i$ is thus modifying only $|S_i|$ blocks of $k$ entries. 
Moreover, $f_i$ is \emph{minimizing} the sum of the risks of agents in $S_i$.  Hence, ignoring for now the constraint that the sum of the positions remain constant, $f_i$ is performing a block coordinate descent step of the surplus function $\Phi$ on this block of coordinates.

\subsection{Randomized coordinate descent}

Several randomized coordinate descent methods have appeared in the literature recently, with increasing levels of sophistication.
While earlier methods focused on updates which only modified disjoint blocks of coordinates~\citep{nesterov2012efficiency,richtarik2014iteration}, more recent methods allow for more general configurations, such as overlapping blocks~\citep{necoara2013random,necoara2014random,reddi2014large}.  In fact, these last three methods are closest to what we study here; the authors consider an objective which decomposes as the sum of convex functions on each coordinate, and study coordinate updates which follow a graph structure, all under the constraint that coordinates sum to $0$.  Despite the similarity of these methods to our trade dynamics, we require even more general updates, as we allow coordinate $i$ to correspond to arbitrary subsets $S_i$.

Fortunately, as we will see, there is a common technique underpinning the five papers mentioned above, and by viewing this technique abstractly, we can generalize it to our setting.  Roughly speaking, the recipe is as follows:
\begin{enumerate}\setlength{\itemsep}{0pt}
\item Derive a quadratic upper bound via Lipschitz continuity;
\label{item:lipschitz}
\item Minimize this upper bound to obtain the update step; \label{item:update}
\item Pick a norm based on the update which captures the expected progress per iteration;
\label{item:norm}
\item Use the definition of the dual norm and the convexity of the objective to relate this progress to the optimality gap and a global notion of distance (the function $\Rc^2$ below);
\label{item:gap-bound}
\item Chain the per-iteration progress bounds into a convergence rate.
\label{item:chain-bounds}
\end{enumerate}

We now follow this recipe to present and analyze a general randomized coordinate descent method, Algorithm~\ref{alg:coord}, for a convex objective $F$ which performs updates on arbitrary subspaces.  We will represent these subspaces as matrices $\{A_i\in\reals^{n\times k_i}\}$, where an update in coordinate $i$ is constrained to be in the image space of $A_i$.  In other words, if $x^{t+1} \gets x^t + d$, we require $d \in \im(A_i)$.  We believe that our analysis can be used to recover the smooth-objective results from the five papers mentioned above.

\raf{FOR LATER: I believe the analysis is much simpler to label $P_i := A_i\aip$ and just say that some orthogonal projection matrices $P_i\in\reals^{n\times n}$ are given (perhaps we should still stick to $A_i$ or even $U_i$ to avoid confusion with $p_i$).  This eliminates the need to (a) write $\aip$ everywhere, and (b) keep track of the dimensions $k_i$ of the various subspaces.}

\begin{algorithm}
  \caption{Randomized Coordinate Descent, adapted from~\cite[Alg. 3]{richtarik2014iteration}}
  \label{alg:coord}
  \begin{algorithmic}[1]
    \Require{Convex function $F:\reals^n\to\reals$, initial point $x^0\in\reals^n$, matrices $\{A_i\in\reals^{n\times k_i}\}_{i=1}^m$, smoothness parameters $\{L_i\}_{i=1}^m$, distribution $p\in\Delta_m$}
    \For{iteration $t$ in $\{0,1,2,\cdots\}$}
    \State Sample $i$ from $p$
    \Let{$x^{t+1}$}{$x^t - \tfrac 1 {L_i} A_i \aip \,\nabla F(x^t)$} \label{alg:coord:update}
    \EndFor
  \end{algorithmic}
\end{algorithm}

We will assume that $F$ is $L_i$-smooth with respect to the image space of the $A_i$; this is step~\ref{item:lipschitz} of our recipe.  Precisely, we require the existence of constants $L_i$ such that for all $y\in\im(A_i)$,
\begin{equation}
  \label{eq:strong-smooth}
  F(x+y) \leq F(x) + \inner{\nabla F(x),y} + \tfrac {L_i} 2 \|y\|_2^2~,
\end{equation}
and refer to this condition as $F$ being \emph{$L_i$-$A_i$-smooth}.  Note that as prescribed by step~\ref{item:update} of our general approach, minimizing this bound over all $x'$ for $y = A_i x'$ yields the update on line~\ref{alg:coord:update} of the algorithm.  \raf{With projections, this is just observing that the minimizer of $\|z-y\|^2$ over $y\in\im(P_i)$ is simply $y = P_i z$, which holds by basic properties of orthogonal projections.}

On to step~\ref{item:norm}, we now introduce a seminorm $\|\cdot\|_A$ which will measure the progress per iteration of Algorithm~\ref{alg:coord}:
\begin{equation}
  \label{eq:coord-norm}
  \|x\|_A := \left(\sum_{i=1}^m \frac {p_i} {L_i} \|A_i \aip x\|_2^2\right)^{1/2},
\end{equation}
where $M^+$ denotes the Moore-Penrose pseudoinverse of $M$.
Note that this is a Euclidean seminorm $\|x\|_A = \inner{Ax,x}$ with $A = \sum \frac {p_i}{L_i} A_i\aip$.  By allowing this to be a seminorm, we can easily capture linear constraints, such as $\sum_i x_i = c$ for some constant $c$; to achieve this, simply ensure that it holds for $x^0$ and that $\ones \in \ker A_i$ for all $i$.  Also, in contrast to~\cite{nesterov2012efficiency,richtarik2014iteration}, we do not assume that the matrices $A_i$ have disjoint images.

\newcommand{\Farg}{F^{\mathrm{arg}}}
\newcommand{\Fmin}{F^{\mathrm{min}}}

Finally, for step~\ref{item:gap-bound} of our recipe, we will need the dual norm of $\|\cdot\|_A$, from which we may define the distance function we need.  Let $X(A) := \{x^0 + Ay:y\in\reals^n\}$ denote the optimization domain, and $\Fmin := \min_{x\in X(A)} F(x)$ and $\Farg := \argmin_{x\in X(A)} F(x)$ denote the minimum and minimizers of $F$, respectively.
\begin{align}
  \|y\|_A^* &:=
  \begin{cases}
    \inner{A^+ y,y}^{1/2} & \text{if } y\in\im(A)\\
    \infty & \text{otherwise}.
  \end{cases}
  \label{eq:risknets-dual-norm}
  \\
  \Rc^2(x_0) &:= \max_{x\in X(A):F(x)\leq F(x^0)} \; \max_{x^* \in \Farg} \|x-x^*\|_A^{*\;2}~.
  \label{eq:R}
\end{align}
One can indeed check that $\|\cdot\|_A^*$ is the dual norm of $\|\cdot\|_A$, in the sense that $\left(\frac 1 2 \|\cdot\|_A^2\right)^* = \frac 1 2 \|\cdot\|_A^{*\,2}$.

We are now ready to prove an $O(1/t)$ convergence rate for Algorithm~\ref{alg:coord}.  Our analysis borrows heavily from~\cite{richtarik2014iteration} and~\cite{necoara2014random}; we give the full proof in Appendix~\ref{sec:risknets-proof-coord} for completeness.
\begin{theorem}
  \label{thm:coord}
  Let $F$, $\{A_i\}_i$, $\{L_i\}_i$, $x^0$, and $p$ be given as in Algorithm~\ref{alg:coord}, with the condition that $A_i$ are full rank and
$F$ is $L_i$-$A_i$-smooth for all $i$.
Then\vspace{-6pt}
  \begin{equation}
    \label{eq:risknets-1}
    \E{}{F(x^t) - \Fmin} \leq \frac{2\Rc^2(x^0)}{t}~.
  \end{equation}
\end{theorem}

To illustrate the power of Theorem~\ref{thm:coord}, we show in Appendix~\ref{sec:risknets-nesterov-bounds} how to leverage results from spectral graph theory to recover results in the literature for specific graphs, each of which correspond to special cases of the subspaces $\{A_i\}_i$.  For now, we simply apply our results to the risk network framework.

\subsection{Application to trade dynamics}

To apply Theorem~\ref{thm:coord} to our setting, we will simply take $F = \Phi$, viewed as a function on $\reals^{N k}$, and construct subspaces which correspond to the trade function subsets $S_i$.  (Recall that $\pos = \reals^k$.)  For each $S_i$, pick some $j^* \in S_i$ and let $B_i \in \reals^{N\times (|S_i|-1)}$ with columns $e_j - e_{j^*}$ for $j\in  S_i\setminus\{j^*\}$, where the $e_j$ are standard column vectors with $1$ in position $j$ and $0$ elsewhere.  This enforces the constraint that $B_i y$ sums to $0$ for all $y\in\reals^{|S_i|-1}$.  To turn this into a matrix of the correct dimensions, we will merely take the Kronecker product with the $k\times k$ identity matrix, $I_{ k}$.  Putting this all together, our setting can be expressed as
\begin{align*}
  F:\reals^{N k}\to\reals;\;\;  F(x) = \Phi(x)\\
  A_i\in\reals^{N k\times(|S_i|-1) k};\;\; A_i = B_i\otimes\!I_{ k}~.
\end{align*}

Of course, to apply Theorem~\ref{thm:coord}, we will need $\Rc^2(x^0)$ to be finite; otherwise the bounds given are meaningless.  This is easily obtained, as the set $\{x:F(x)\leq F(x^0)\}$ is bounded (and convex), and 
for all $x^*\in \Farg$ and $x\in X(A)$ we have $x - x^* \in \im(A)$, so by eq.~\eqref{eq:risknets-dual-norm} $\|x-x^*\|_A^*$ is also bounded.

\begin{theorem}
  \label{thm:dynamics-rates}
  For any connected trade dynamic, we have $\E{}{\Phi(r^t)} = O(1/t)$.
\end{theorem}
\begin{proof}
  As our risks are differentiable, their conjugates are strictly convex functions, and hence strongly convex as $\Pi$ is compact.  We conclude that we have some $\{\sigma_i\in\reals\}_i$ such that $\risk_i$ is $\sigma_i$-smooth.  Now taking $L_i = \max_{j \in S_i} \sigma_j$, one can see that $F$ is $L_i$-$A_i$-smooth for all $i$ by eq.~\eqref{eq:strong-smooth}. \rfnote{I'm still not positive about this, but okay for now --- it's at least obvious that $F$ is $L$-smooth for some $L$.}

  Now note that Algorithm~\ref{alg:coord} carries no state aside from $x^t$, and hence the analysis depends on the drop in the objective function per step.  In other words, given $x^t$, the analysis holds for any algorithm satisfying $F(x^{t+1})  \leq F(x^t - A_i\aip \nabla F(x^t))$.  As trade dynamics actually minimize $F(x^t - A_i y)$ over all $y$, this property trivially holds, and thus the result of Theorem~\ref{thm:coord} applies.
\end{proof}

Amazingly, Theorem~\ref{thm:dynamics-rates} holds for \emph{all} connected trade dynamics, as they each minimize the surplus in whichever $S_i$ is chosen, and that is enough for the bounds from Theorem~\ref{thm:coord} to apply.  In fact, it is more than enough: as Theorem~\ref{thm:coord} holds even for gradient updates as in Algorithm~\ref{alg:coord}, the rates extend to less efficient trade dynamics, as long as the drop in surplus is at least as large as, or even within a constant factor of, the gradient update in Algorithm~\ref{alg:coord}.  This suggests that our convergence results are robust with respect to the model of rationality one employs; if agents have bounded rationality and cannot compute positions which would exactly minimize their risk, but instead approximate it within a constant factor of the gradient update, the rate remains $O(1/t)$.

\section{Application to Prediction Markets}
\label{sec:risknets-appl-pred-mark}

Our analysis was motivated in part by work that considered the equilibria of prediction markets with specific models of trader behavior: traders as risk minimizers~\citep{Hu:2014}; and traders with exponential utilities and beliefs from exponential families~\citep{Abernethy:2014}.
In both cases, the focus was on understanding the properties of the market at convergence, and questions concerning whether and how convergence happened were left as future work. 
We now explain how this earlier work can be seen as a special case of our analysis with an appropriate choice of network structure and dynamics. In doing so we also generalize several earlier results.

Following~\cite{abernethy2013efficient}, a cost function-based prediction market consists of a collection of $k$ outcome-dependent \emph{securities} $\{\phi(\cdot)_i\}_{i=1}^k$ that pay $\phi(\oo)_i$ dollars should outcome $\oo \in \outcomes$ occur. 
A \emph{market maker} begins with an initial position $r^0 \in \pos$, the \emph{liability vector}, and a 
\emph{cost function} $C : \pos \to \R$.
A trader who wishes to purchase a bundle of securities $r \in \pos$ is
charged $\price(r) := C(r^t + r) - C(r^t)$ by the market maker which then updates its liability to $r^{t+1} = r^t + r$.  The desirable properties for cost functions are quite different from those of risk measures (e.g. information incorporating, arbitrage-free), yet as observed by \cite{Hu:2014}, the duality-based representation of cost functions  is essentially the same as the one for risk measures (compare Theorem~\ref{thm:cvx-risk-dual} and \cite[Theorem 5]{abernethy2013efficient}).  In essence then, cost functions \emph{are} risk measures, though because liability vectors measure losses and position vectors measure gains, we simply have $\risk_C(r) = C(-r)$.

In the prediction market of \cite{Hu:2014}, agents have risk measures $\risk_i$ and positions $r_i$. A trade of $r$ between such and agent a market maker with cost function $C$ and position $r^t$ makes the agent's new risk
$\risk_i(r_i + r - \price(r)\cdot\cash )$ since the market maker charges $\price(r)$ dollars for $r$.  Similarly, one can check that the market maker's risk remains constant for all trades of this form.

An agent minimizing its risk implements the \emph{trading function} (Definition~\ref{def:risknets-trade-func}) $f : (-r^t, r_i) \mapsto (-r^t-r, r_i+r)$ since
$\min_r \risk_i(r_i + r - (C(r^t + r) - C(r^t))\cdot\cash)
= \min_r \risk_i(r_i + r) + \risk_C(-r^t - r)$ by cash invariance of $\risk_i$, guaranteeing the surplus between the agent and market maker is zero.
Thus, one could think of agents in a risk-based prediction market as residing on a star graph, with the market maker in the center.  By Theorem~\ref{thm:dynamics-equilibrium}, any trade dynamic which includes every agent with positive probability will converge, and Theorem~\ref{thm:dynamics-rates} gives an $O(1/t)$ rate of convergence.



An important special case is where agents all share the same base risk measure $\risk$, but to different degrees $b_i$ which intuitively correspond to a level of \emph{risk affinity}.  Specifically, let $\risk_i(r) = b_i \risk(r/b_i)$, where a higher $b_i$ corresponds to a more risk-seeking agent.\footnote{Note however that agents are still risk-averse; only in the limit as $b\to\infty$ do the traders become risk-neutral.}  As we now show, the market equilibrium gives agent $i$ a share of the initial sum of positions $r^0$ proportional to his risk affinity, and the final ``consensus'' price of the market is simply that of a scaled version of $r^0$.
\begin{theorem}
  \label{thm:risknets-perspectives}
  Let $\risk$ be a given risk measure, and for each agent $i$ choose an initial position $r_i^0\in\pos$ and risk defined by $\risk_i(r_i) = b_i\risk(r_i/b_i)$ for some $b_i>0$.  Let $r^0 = \sum_i r_i^0$, and define $r\in\pos^N$ by $r_i = b_i r^0 / \textstyle\sum_j b_j$.  Then $r$ is the unique point up to zero-sum cash transfers such that $\Phi(r) = 0$.  Moreover, $r$ satisfies for all $i$,\vspace{-3pt}
  \begin{equation}
    \label{eq:risknets-2}
    \nabla \risk_i(r_i) = \nabla \risk\left(r^0/\textstyle\sum_j b_j\right)~.
  \end{equation}
\end{theorem}
\begin{proof}
  Note that $\nabla\risk_i(r_i) = \nabla\risk(r_i/b_i) = \nabla\risk(r^0/\sum_j b_j)$.  By the proof of Theorem~\ref{thm:risknets-trade-function}, $r$ must then satisfy $\Phi(r)=0$, and is the unique such point up to cash transfers.
\end{proof}

This result generalizes those in \S5 of \cite{Abernethy:2014}, where traders are assumed to maximize an expected utility of the form $U_b(w) = -b \exp(-w/b)$ under beliefs drawn from an exponential family with sufficient statistic given by the securities $\phi$.
The above result shows that exactly the same weighted distribution of positions at equilibrium occurs for \emph{any} family of risk-based agents, not just those derived from exponential utility via certainty equivalents
\citep{Ben-Tal:2007}.
In addition, this generalization shows that the agents need not have exponential family beliefs: their positions $r_i$ act as general natural parameters, and $1/b_i$ acts as a general measure of risk aversion.
Finally, this connection also means our analysis applies to their setting, addressing their future work on dynamics and convergence.


\rfnote{Added for arXiv version}
\paragraph{Remarks.}

In Section~\ref{sec:setting}, we observed that one could think of $\infc_i\risk_i$ as the ``market risk''; we now have enough context to understand the impact of this idea, by considering interactions \emph{between} markets.  Consider an arbitrary connected graph $G$ on $N$ ``meta-agents'', and for each of these agents, attach a collection of new ``child'' agents, calling the combined graph $H$; hence, each node in $H$ which came from $G$ is the center of its own star graph (see Figure~\ref{fig:meta-market}).  One can think of this setting as a network of market makers, each with disjoint trading populations.  By the associativity of the infimal convolution, and Theorem~\ref{thm:dynamics-equilibrium}, the equilibrium of the combined graph $H$ is the same as that of $G$, if we replace the risk of each meta-agent by the infimal convolution of the risks in its star graph.
Another instantiation of this idea would be to build a hierarchical market, or ``deep market'', corresponding to a massive tree, where each node serves as a market maker for the nodes below, but acts as a trader in the market above.  Again, under our model, the aggregation properties of such a hierarchical market would be exactly the same as the flattened market, where all agents interacted directly with the root market maker.

\begin{figure}
  \centering
  \begin{tikzpicture}[node distance=3cm,%
    meta/.style={outer sep=4pt,draw=black,shape=circle},%
    child/.style={outer sep=0pt,inner sep=1pt,draw=gray,shape=circle},%
    bend angle=30,->,>=stealth']
    \foreach \i in {1,...,6}
    {
      \node[meta] (\i) at (-120-\i*60:2) {$M_{\i}$};
      \foreach \j in {1,...,5}
      {
        \path (\i) --++(-120-\i*60+20*\j-60:2) node[child] (\i-\j) {$a^{\i}_{\j}$};
        \path (\i-\j) edge (\i);
      }
    }
    \path (1) edge (4) edge (5) edge (6);
    \path (2) edge (3) edge (4) edge (6);
    \path (3) edge (5) edge (6) edge (1);
    \path (4) edge (1) edge (2);
    \path (5) edge (6) edge (2) edge (4);
    \path (6) edge (2) edge (3);
  \end{tikzpicture}%
  \caption{A meta market.}
  \label{fig:meta-market}  
\end{figure}
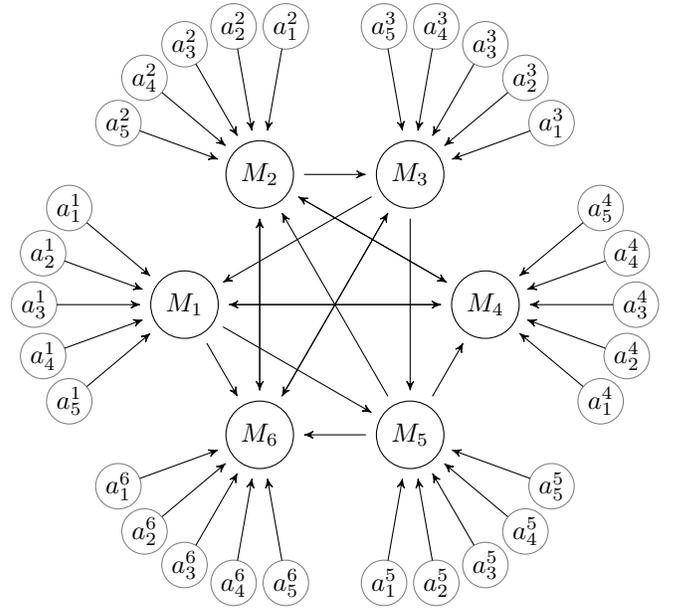

Another interpretation of the idea of meta-agents is through the lens of Coase's theory of the firm~\citep{Coase:1937}. 
Firms, according to Coase, arise when transaction costs make market coordination inefficient, allowing agents to coordinate without recourse to a price mechanism.
Our risk network model could be extended by introducing a fixed transaction cost for each trade along an edge in a network.
If, in the meta-agent example described above, the average surplus per agent in a star graph was comparable to the transaction cost then it is advantageous for the group to form a ``firm'' where agents agree to aggregate their positions and risk since trades along the edges of the star graph will deplete each agent's share of the surplus.
All agents then stand to gain from the meta-agent's interaction with the rest of the network without incurring the transaction costs required to redistribute those gains.
Given an initially unstructured collection of agents, one can imagine a network forming to offset transaction costs. 
Different groupings of agents into firms and the placement of edges between them would have an effect on how much of the global surplus would be taken by agents versus how much would be lost to transaction costs.
The specifics of this sort of model is left to future work.


\begin{figure}[!ht]
\centering
\vspace{-10pt}
\includegraphics[width=\columnwidth]{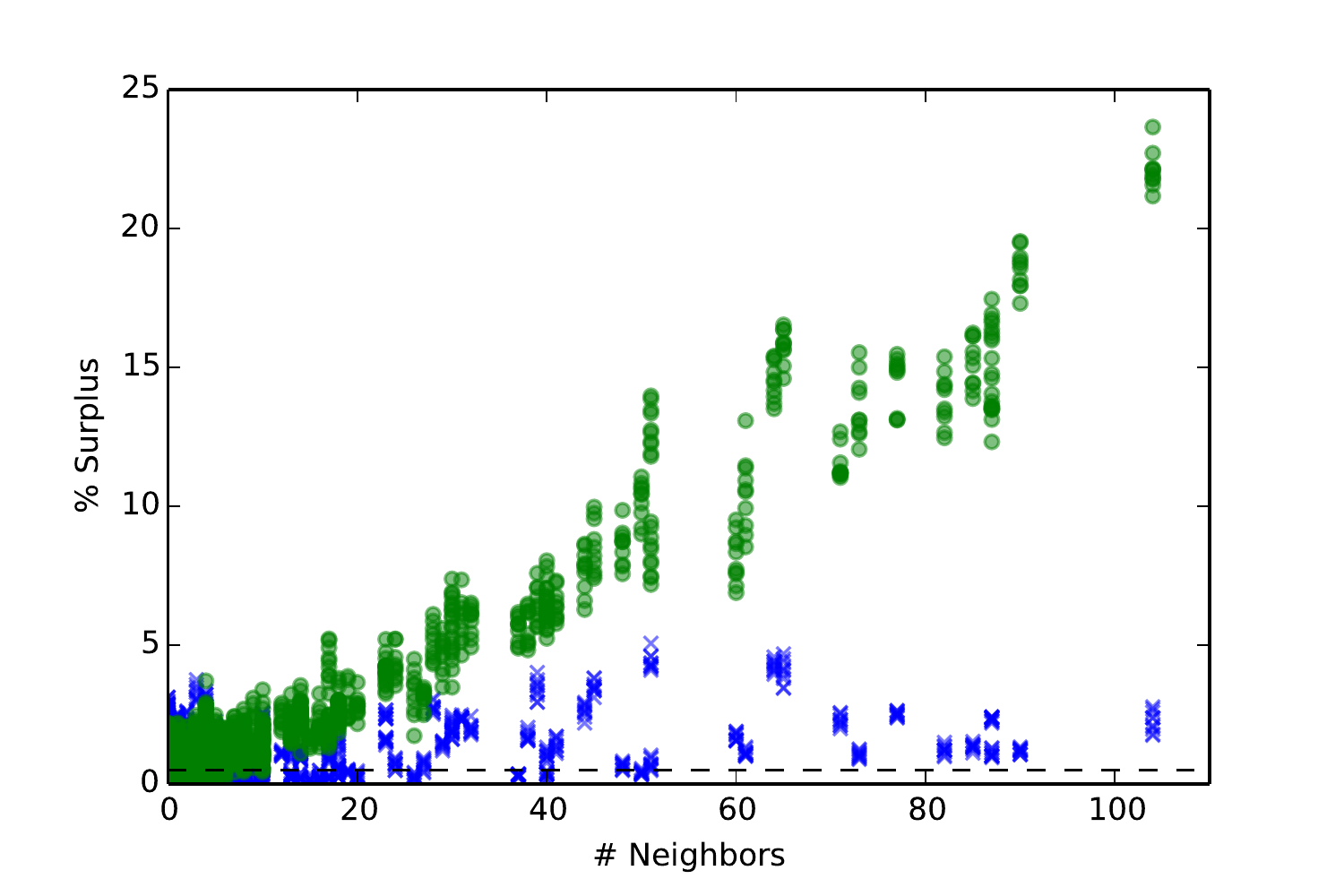}
\caption{Percentage of captured surplus per trader vs. number of trading neighbors for fair edge dynamic (green circles) and fair node dynamic (blue crosses). The dashed black line shows the fair distribution for 200 agents (0.5\%).\label{fig:drop_vs_degree}\vspace{-5pt} }
\end{figure}

\section{Experimental Results}

The theoretical results above show that the minimal value of a risk
network is independent of its topology and the dynamic used to achieve 
convergence.
However, due to the cash invariance of the traders' risk measures there are many
different trades that can reduce the network's surplus to zero.
If a fair dynamic is used to find the optimal trade amongst all traders in a network then convergence will occur in one step and, by the definition, the drop in risk for all the traders will be equal.
When traders are only allowed to trade with the neighbors in a network it is less obvious what effect the dynamics will have on the final distribution of risk at convergence.

To better understand the effect of network topology and dynamics on the
redistribution of risk, we implemented a simulation of network 
trading dynamics in Python. 
The package \texttt{networkx} is used to generate random scale-free networks of 200 agents.
We opted to study scale-free networks as these have properties similar
to naturally occurring networks (\eg, power law distributed vertex 
degrees).
Each agent in our simulated networks have entropic risk and uniformly randomly drawn positions from $[-50,50]^3$.
Node (respectively, edge) dynamics are implemented by choosing
a network vertex (resp., edge) uniformly at random and then finding
and executing the optimal trade between the neighbours of the
selected vertex (resp., endpoints of the edge).
Finding optimal trades between a collection of agents is implemented
using SciPy's \texttt{optimize} method.
In both experiments, 10 different networks (\ie, with different initial positions and structure) were each run 10 times with the edge dynamic and with the node dynamic, giving a total of 100 trials with 200 agents for each dynamic. Each trial was run for a maximum of 1000 steps of the dynamic.

Figure~\ref{fig:drop_vs_degree} shows how the number of trading neighbors an agent has affects the proportion of the surplus that agent takes once the network converges. For each dynamic, each of the $100\times 200$ points shows a single agent's degree and the percentage of the surplus it took at the end of each trial.
The results clearly show the strong influence of agent degree on its ability to minimize its risk under edge dynamics.
The effect is much weaker for node dynamics.
One possible explanation for this is that
high degree agents are selected less frequently under node dynamics, but also have to share the surplus with many more nodes.  We give further plots in the Appendix.

\section{Conclusions}

We have developed a framework to analyze arbitrary networks of risk-based agents, giving a very general analysis of convergence and rates, and addressing open issues in both \cite{Hu:2014} and \cite{Abernethy:2014}.  We view this as a foundation, which opens more questions than it answers.  For example, can we improve the asymptotic rates of convergence?  One potential technique would be to show that trading never leaves a bounded region, and carefully applying bounds for strongly convex functions (modulo the $\cash$ direction), which could give a rate as fast as $O(1/2^t)$.  An even broader set of questions has to do with the distribution of risk --- how does the network topology effect the outcome on the agent level?  As our experiments show, even local properties of the network may have a strong effect on the final distribution of risks, and understanding this relationship is a very interesting future direction.

\subsection*{Acknowledgments}
We would like to thank Matus Telgarsky for his generous help, as well as the lively discussions with, and helpful comments of, S\'ebastien Lahaie, Miro Dud\'ik, Jenn Wortman Vaughan, Yiling Chen, David Parkes, and Nageeb Ali.

\bibliographystyle{aaai}
\bibliography{risknets,diss}

\begin{thebibliography}{}

\bibitem[\protect\citeauthoryear{Abernethy \bgroup et al\mbox.\egroup
  }{2014}]{Abernethy:2014}
Abernethy, J.; Kutty, S.; Lahaie, S.; and Sami, R.
\newblock 2014.
\newblock Information aggregation in exponential family markets.
\newblock In {\em Proceedings of the fifteenth ACM conference on Economics and
  computation},  395--412.
\newblock ACM.

\bibitem[\protect\citeauthoryear{Abernethy, Chen, and
  Vaughan}{2013}]{abernethy2013efficient}
Abernethy, J.; Chen, Y.; and Vaughan, J.~W.
\newblock 2013.
\newblock Efficient market making via convex optimization, and a connection to
  online learning.
\newblock {\em {ACM} Transactions on Economics and Computation} 1(2):12.

\bibitem[\protect\citeauthoryear{Artzner \bgroup et al\mbox.\egroup
  }{1999}]{artzner1999coherent}
Artzner, P.; Delbaen, F.; Eber, J.-M.; and Heath, D.
\newblock 1999.
\newblock Coherent measures of risk.
\newblock {\em Mathematical finance} 9(3):203--228.

\bibitem[\protect\citeauthoryear{Ben-Tal and Teboulle}{2007}]{Ben-Tal:2007}
Ben-Tal, A., and Teboulle, M.
\newblock 2007.
\newblock An old-new concept of convex risk measures: The optimized certainty
  equivalent.
\newblock {\em Mathematical Finance} 17(3):449--476.

\bibitem[\protect\citeauthoryear{Burgert and
  R{\"u}schendorf}{2006}]{Burgert:2006}
Burgert, C., and R{\"u}schendorf, L.
\newblock 2006.
\newblock On the optimal risk allocation problem.
\newblock {\em Statistics \&amp; decisions} 24(1/2006):153--171.

\bibitem[\protect\citeauthoryear{Coase}{1937}]{Coase:1937}
Coase, R.~H.
\newblock 1937.
\newblock The nature of the firm.
\newblock {\em Economica} 4(16):386--405.

\bibitem[\protect\citeauthoryear{de Abreu}{2007}]{de_abreu2007old}
de~Abreu, N. M.~M.
\newblock 2007.
\newblock Old and new results on algebraic connectivity of graphs.
\newblock {\em Linear algebra and its applications} 423(1):53--73.

\bibitem[\protect\citeauthoryear{F{\"o}llmer and Schied}{2004}]{Follmer:2004}
F{\"o}llmer, H., and Schied, A.
\newblock 2004.
\newblock Stochastic finance, volume 27 of de gruyter studies in mathematics.

\bibitem[\protect\citeauthoryear{Hiriart-Urruty and
  Lemar\'{e}chal}{1993}]{hiriart1993grundlehren}
Hiriart-Urruty, J., and Lemar\'{e}chal, C.
\newblock 1993.
\newblock Grundlehren der mathematischen wissenschaften.
\newblock {\em Convex Analysis and Minimization Algorithms II} 306.

\bibitem[\protect\citeauthoryear{Hu and Storkey}{2014}]{Hu:2014}
Hu, J., and Storkey, A.
\newblock 2014.
\newblock Multi-period trading prediction markets with connections to machine
  learning.
\newblock In {\em Proceedings of the 31st International Conference on Machine
  Learning (ICML)}.

\bibitem[\protect\citeauthoryear{Meucci}{2009}]{meucci2009risk}
Meucci, A.
\newblock 2009.
\newblock {\em Risk and Asset Allocation}.
\newblock Springer Science \& Business Media.

\bibitem[\protect\citeauthoryear{Mohar}{1991}]{mohar1991laplacian}
Mohar, B.
\newblock 1991.
\newblock The laplacian spectrum of graphs.
\newblock In {\em Graph Theory, Combinatorics, and Applications}.
\newblock Citeseer.

\bibitem[\protect\citeauthoryear{Necoara, Nesterov, and
  Glineur}{2014}]{necoara2014random}
Necoara, I.; Nesterov, Y.; and Glineur, F.
\newblock 2014.
\newblock A random coordinate descent method on large-scale optimization
  problems with linear constraints.
\newblock {\em Technical Report}.

\bibitem[\protect\citeauthoryear{Necoara}{2013}]{necoara2013random}
Necoara, I.
\newblock 2013.
\newblock Random coordinate descent algorithms for multi-agent convex
  optimization over networks.
\newblock {\em Automatic Control, IEEE Transactions on} 58(8):2001--2012.

\bibitem[\protect\citeauthoryear{Nesterov}{2012}]{nesterov2012efficiency}
Nesterov, Y.
\newblock 2012.
\newblock Efficiency of coordinate descent methods on huge-scale optimization
  problems.
\newblock {\em SIAM Journal on Optimization} 22(2):341--362.

\bibitem[\protect\citeauthoryear{Othman and
  Sandholm}{2011}]{othman2011liquidity-sensitive}
Othman, A., and Sandholm, T.
\newblock 2011.
\newblock Liquidity-sensitive automated market makers via homogeneous risk
  measures.

\bibitem[\protect\citeauthoryear{Reddi \bgroup et al\mbox.\egroup
  }{2014}]{reddi2014large}
Reddi, S.; Hefny, A.; Downey, C.; Dubey, A.; and Sra, S.
\newblock 2014.
\newblock Large-scale randomized-coordinate descent methods with non-separable
  linear constraints.
\newblock {\em arXiv preprint arXiv:1409.2617}.

\bibitem[\protect\citeauthoryear{Richt{\'a}rik and
  Tak{\'a}{\v{c}}}{2014}]{richtarik2014iteration}
Richt{\'a}rik, P., and Tak{\'a}{\v{c}}, M.
\newblock 2014.
\newblock Iteration complexity of randomized block-coordinate descent methods
  for minimizing a composite function.
\newblock {\em Mathematical Programming} 144(1-2):1--38.

\bibitem[\protect\citeauthoryear{Rockafellar}{1997}]{Rockafellar:1997}
Rockafellar, R.
\newblock 1997.
\newblock {\em Convex analysis}.
\newblock Princeton University Press.

\end{thebibliography}

\clearpage
\appendix

\section{Proof of Theorem~\ref{thm:coord}}

Before giving the proof, we note that the result in~\cite[Thm 11]{richtarik2014iteration} also holds for general Euclidean norms $\|\cdot\|_{(i)}$.  We leave out such extensions as ultimately the only change is in the update step (by leveraging e.g.~\cite[Lemma 10]{richtarik2014iteration} instead of our pseudoinverse update) and the form of the dual norm.  Lemma~\ref{lem:w-norm} below verifies that $\|\cdot\|_A$ is still a seminorm in these cases.

\label{sec:risknets-proof-coord}
\begin{proof}[Proof of Theorem~\ref{thm:coord}]
  To begin, suppose subspace $i$ is chosen at step $t$.  Letting $z = \frac{1}{L_i}\aip\nabla F(x^t)$ and $y = A_iz \in \im(A_i)$, the drop in the objective can be bounded using eq.~\eqref{eq:strong-smooth},
  \begin{equation}
    \label{eq:risknets-3}
    F(x^t) - F(x^t - y)
    \geq \inner{\nabla F(x^t),y}-\frac{L_i}2 \|y\|_2^2~.
  \end{equation}
  By properties of the Moore-Penrose pseudoinverse, we have
  \begin{align*}
    & \argmax_{z\in\reals^n} \inner{\nabla F(x^t),A_i z}-\frac{L_i}2 \|A_i z\|_2^2
    \\
    &= \argmin_{z\in\reals^n} \|A_i z - \nabla F(x^t)\|_2 \;= \aip z,
  \end{align*}
  but we also have that
  \begin{align*}
    & \max_{z\in\reals^n} \inner{\nabla F(x^t),A_i z}-\frac{L_i}2 \|A_i z\|_2^2
    \\
    &= \max_{z\in\reals^n} \inner{A_i^\tr \nabla F(x^t),z}-\frac{L_i}2 \inner{A_i^\tr A_i z,z}
    \\
    &= \frac 1 {2 L_i} \inner{(A_i^\tr A_i)^{-1} A_i^\tr \nabla F(x^t),A_i^\tr \nabla F(x^t)}
    \\
    &= \frac 1 {2 L_i}\|A_i\aip\nabla F(x^t)\|_2^2~,
  \end{align*}
  where we used the fact that $A_i$ is full rank and the identities $A^+ = (A_i^\tr A_i)^{-1}A_i^\tr$ and $A_i\aip A_i = A_i$.  Putting these together with eq.~\eqref{eq:risknets-3} and the fact that $x^{t+1} = x^t - y$ when $i$ is chosen, we have
  \begin{equation}
    \label{eq:risknets-4}
    F(x^t) - F(x^{t+1})
    \geq \frac 1 {2 L_i}\|A_i\aip\nabla F(x^t)\|_2^2~.
  \end{equation}
  Now looking at the expected drop in the objective, we have
  \begin{align*}
    F(x^t) - \E{}{F(x^{t+1})|x^t}
    &\geq \sum_{i=1}^m p_i \frac 1 {2 L_i}\|A_i\aip\nabla F(x^t)\|_2^2\\ & = \frac 1 2 \|\nabla F(x^t)\|_A^2~.
  \end{align*}
  
  To complete step~\ref{item:gap-bound} of our recipe and relate our per-round progress to the gap remaining, we observe that
  \begin{align*}
 F(x^t&) - \Fmin\\
    &\leq \max_{x^*\in\argmin_x F(x)} \inner{\nabla F(x^t),x^*-x^t} \\
    &\leq \max_{x^*\in\argmin_x F(x)} \|\nabla F(x^t)\|_A \;\|x^*-x^t\|_A^* \\
    &\leq \|\nabla F(x^t)\|_A \;\max_{x^*\in\argmin F} \; \max_{x:F(x)\leq F(x^0)} \|x^*-x\|_A^* \\
    &= \|\nabla F(x^t)\|_A \;\Rc(x^0)~,
  \end{align*}
  where we used convexity of $F$, the definition of the dual norm, the fact that $F(x^t)$ is non-increasing in $t$, and finally the definition of $\Rc$.  We now have $F(x^t) - \E{}{F(x^{t+1})|x^t} \geq (F(x^t)-\Fmin)/(2\Rc^2(x^0))$.  The remainder of the proof follows an argument of~\cite{necoara2014random} by analyzing $\Delta_t = \E{}{F(x^t)-\Fmin}$.  From the last inequality we have $\Delta_{t+1} \leq \Delta_t - \Delta_t^2 / 2\Rc^2(x^0)$, and since $\Delta_{t+1} \leq \Delta_t$, this gives $\Delta_t^{-1} \leq \Delta_{t+1}^{-1} - (2\Rc^2(x^0))^{-1}$.  Summing these inequalities gives the result.
\end{proof}

\label{sec:risknets-norm-lemma}
\begin{lemma}
  \label{lem:w-norm}
  Let seminorms $\{\|\cdot\|_{(i)}\}_{i=1}^m$ and positive weights $\{w_i\}_{i=1}^m$ be given, and define the function $\| \cdot \|_W : \reals^n\to\reals$ by
  \begin{equation}
    \label{eq:w-norm}
     \| x \|_W = \left(\sum_{i=1}^m w_i \|x\|_{(i)}^2\right)^{1/2}~.
  \end{equation}
  Then $\|\cdot\|_W$ is a seminorm.  It is additionally a norm if and only if $\| x \|_{(i)} = 0$ holds for all $i$ only when $x=0$.
\end{lemma}
\begin{proof}
  First, note that we may fold the weights into the seminorms, $\|x\|_{(i)}' := \|\sqrt{w_i}\,x\|_{(i)}$, so we can assume $w_i=1$ for all $i$ without loss of generality.
  Let $\varphi:\reals^n\to\reals^m$ be given by $\varphi(x)_i = \|x\|_{(i)}$.  Then $\|x\|_W = \|\varphi(x)\|_2$.
  \begin{itemize}
  \item Absolute homogeneity.  First observe that $\varphi(\alpha x) = |\alpha| \varphi(x)$ by homogeneity of the $\|\cdot\|_{(i)}$.  Then $\|\alpha x\|_W = \||\alpha| \varphi(x)\|_2 = |\alpha| \|x\|_W$.
  \item Subadditivity.  We first recall the fact that if $x_i \geq y_i$ for all $i$, then $\|x\|_2 \geq \|y\|_2$.  Combining this fact with subadditivity of the $\|\cdot\|_{(i)}$ and then of $\|\cdot\|_2$, we have
    \begin{align*}
      \|x+y\|_W &= \|\varphi(x+y)\|_2 \leq \|\varphi(x)+\varphi(y)\|_2\\ &\leq \|\varphi(x)\|_2 + \|\varphi(y)\|_2 = \|x\|_W + \|y\|_W.
    \end{align*}
  \end{itemize}

  We now show the norm condition.  First, we assume $\| x \|_{(i)} = 0$ for all $i$ implies $x=0$; we will show Separation.  We clearly have $\|0\|_W=0$.  By the above, $\|x\|_W = 0$ implies $\|\varphi(x)\|_2 = 0$, yielding $\|x\|_{(i)}=0$ for all $i$ by definiteness of $\|\cdot\|_2$, and hence $x = 0$ by assumption.

  For the converse, observe that any $x\neq 0$ with $\|x\|_{(i)} = 0$ for all $i$ would imply a violation of definiteness, as $\varphi(x) = 0$ and hence $\|x\|_W = \|\varphi(x)\|_2 = \|0\|_2 = 0$.
\end{proof}

\section{Coordinate descent bounds}
\label{sec:risknets-nesterov-bounds}

We show here that by formulating a general algorithm to perform coordinate descent steps along arbitrary subspaces, we can recover existing algorithms via reduction to the computation of a matrix pseudoinverse.  As we will see, for the special case of edge updates in a graph, we can leverage existing results in spectral graph theory to analyze new graphs currently not considered in the literature.

Let us first consider an optimization problem on the complete graph, which picks an edge $(i,j)$ uniformly at random and optimizes in coordinates $i$ and $j$ under the constraint that $x_i^{t+1} + x_j^{t+1} = x_i^{t} + x_j^{t}$.  This corresponds to subspaces $A_{(i,j)} = e_i-e_j$, where $e_i$ is the $i$th standard unit vector, making $A_{(i,j)} A^+_{(i,j)} = \frac 1 2 (e_i-e_j)(e_i-e_j)^\tr$.  Assuming a global smoothness constant $L$, one can calculate
\begin{align*}
  A &= \tfrac 2 {L N (N-1)} \sum_{(i,j)} A_{(i,j)} A_{(i,j)}^+ = \tfrac 1 {L  (N-1)} \left(I - \tfrac 1 N \ones\right)~,\\
  A^+ &= L(N-1) (I - \tfrac 1 N \ones)~,
\end{align*}
where $\ones$ is the $N\times N$ all-ones matrix.  Now as $\im(A)=\ker(\ones)$, this gives
\begin{align}
  \|x\|_A^{*\;2} &= 
  L(N-1)\|x\|_2^2~.\label{eq:risknets-complete-graph}
\end{align}
Similarly, the complete rank-$K$ hypergraph gives $\|x\|_A^{*\;2} = L\frac{N-1}{K-1}\|x\|_2^2$. (Compare to eq. (3.10) and the top of p.21 of~\cite{necoara2014random}.)  Letting $\Cc_0 = 4L\max_{x\in X(A):F(x)\leq F(x^0)} \; \max_{x^* \in \Farg} \|x-x^*\|_2^2$, which is independent of the (hyper)graph as long as it is connected, we thus have a convergence rate of $\tfrac{N-1}{2}\, \Cc_0\, \frac 1 t$ for the complete graph, and more generally $\tfrac{N-1}{2(K-1)}\, \Cc_0\, \frac 1 t$ for the complete $K$-graph.  Henceforth, we will consider the coefficient in front of $\Cc_0$ to be the convergence rate, as all other parameters

The above matrix $A$ is a scaled version of what is known as the \emph{graph Laplacian} matrix; given a graph $G$ with adjacency matrix $A(G)$ and degree matrix $D(G)$ with the degrees of each vertex on the diagonal, the Laplacian is the matrix
\begin{equation}
  \label{eq:risknets-graph-laplacian}
  \Lc = \Lc(G) := D(G) - A(G)~.
\end{equation}
One can check that indeed, $\Lc = 2 \sum_{(i,j)\in E(G)} A_{(i,j)}A_{(i,j)}^+$, meaning $A = \tfrac p {2 L} \Lc$, where $p = 1/|E(G)|$ is the uniform probability on edges.

The graph Laplacian is an extremely well-studied object in spectral graph theory and many other domains, and we can use existing results to establish bounds for more interesting graphs.  In particular, we will be interested in the second smallest eigenvalue of $\Lc$, $\lambda_2(G)$; it is easy to see that the smallest eigenvalue is $\lambda_1(G) = 0$ with eigenvector $\ones$.  The reason for this focus is in the combination of the following two facts: (1) the norm $\inner{Bx,x}^{1/2}$ for symmetric $B$ can be bounded by the maximum eigenvalue of $B$, and (2) the maximum eigenvalue of $B^+$ is equal to the inverse of the smallest nonzero eigenvalue of $B$, provided again that $B$ is symmetric.\footnote{These facts follow from the operator norm and singular-value decomposition for the pseudoinverse, respectively, together with the fact that singular values are eigenvalues for symmetric matrices.}  In particular, the smallest nonzero eigenvalue of $A$ is simply $\tfrac p {2L} \lambda_2(G)$.  Hence, for \emph{any} connected graph $G$, we have
\begin{equation}
  \label{eq:risknets-connectivity-norm-bound}
  \|x\|_A^{*\; 2} \leq 2 L \frac {|E(G)|} {\lambda_2(G)} \|x\|_2^2~.
\end{equation}
Of course, by the above definition of $\Cc_0$ and Theorem~\ref{thm:coord}, this yields the result
\begin{equation}
  \label{eq:risknets-6}
  \E{}{F(x^t) - \Fmin} \leq \frac {|E(G)|} {\lambda_2(G)} \Cc_0 \frac 1 t~,
\end{equation}
showing us how tightly related this eigenvalue is to rate of convergence of Algorithm~\ref{alg:coord}.

As it happens, this second-smallest eigenvalue $\lambda_2(G)$ is called the \emph{algebraic connectivity} of $G$, and is itself thoroughly studied in spectral and algebraic graph theory.  For example, it is known (and easy to check) that $\lambda_2(K_N) = N$, where $K_N$ denotes the complete graph; this together with $|E(K_N)| = N(N-1)/2$ immediately gives eq.~\eqref{eq:risknets-complete-graph}.  In~\cite{de_abreu2007old}, algebraic connectivities are also given for the path on $N$ vertices $P_N$, the cycle $C_N$, the bipartite complete graph $K_{M,K}$ for $K<M$, and the $K$-dimensional cube $B_K$.  Putting these eigenvalues together yields Table~\ref{tab:graphs}.

\begin{table}
  \centering
  \begin{tabular}{lcccc}
    \toprule
    Graph &     $|V(G)|$   &  $|E(G)|$  &    ${\lambda_2(G)}$
    \\\midrule
    $K_N$ & $N$ & $N(N-1)/2$ & $N$
    \\
    $P_N$ & $N$ & $N-1$ & ${2(1\!-\!\cos \!\tfrac \pi N)}$
    \\
    $C_N$ & $N$ & $N$ & ${2(1\!-\!\cos \!\tfrac {2\pi}{N})}$
    \\
    $K_{M,K}$ & $M+K$ & $MK$  & ${K}$
    \\
    $B_K$ & $2^K$ & $K2^{K-1}$ & ${2}$
    \\\bottomrule
  \end{tabular}

\caption{Vertices, edges, and algebraic connectivities for common graphs.}
\label{tab:graphs}
\end{table}
Using the values in Table~\ref{tab:graphs}, we can directly compare the theoretical convergence rates for different graphs.\footnote{Note that of course these are just upper bounds on the true convergence rates.}  For example, the star graph $K_{N-1,1}$ is the natural network for prediction markets in Section~\ref{sec:risknets-appl-pred-mark}; plugging in the values from Table~\ref{tab:graphs} into eq.~\eqref{eq:risknets-connectivity-norm-bound}, we see that, despite its sparsity, the convergence rate for the star graph $(N-1)(1)/(1) = (N-1)$ is within a factor of 2 of the rate for complete graph.  The path and cycle fare much worse, yielding roughly $N/2(N^{-2}/2) = N^3$ as $N$ becomes large (applying the Taylor expansion and ignoring $\pi$ terms).  Finally, an interesting result due to~\citet{mohar1991laplacian} says that for any connected graph on $N$ vertices, we have $\lambda_2(G) \geq 4/(N\mathrm{diam}(G))$ where $\mathrm{diam}(G)$ is the diameter of $G$.  Hence for any graph we certainly have
\begin{equation}
  \label{eq:risknets-5}
  \E{}{F(x^t) - \Fmin} \leq \frac {N\, |E(G)|\, \mathrm{diam}(G)} {4} \, \Cc_0 \frac 1 t~,
\end{equation}
which is a useful bound for sparse graphs of small diameter.

As we have demonstrated above, our general approach to choosing coordinate subspaces combines very naturally with the literature in algebraic and spectral graph theory, yielding a reasonably rich understanding of the convergence rates for various choices of network structure.  In particular, this approach can be used to analyze algorithms for specific networks without needing to start from scratch.  It would be of interest to compute similar bounds for general classes of hypergraphs, to better understand the trade-offs between the convergence rate and the size/connectivity of coordinate subspaces.

\raf{Possibly add later:

For general hypergraphs $\Sc$, we may define the degree matrix $D(\Sc)$ to be the diagonal matrix with $D(\Sc)_{ii} = \#\{S\in\Sc:i\in S\}$, and the ``adjacency'' matrix to be $A(\Sc)_{ij} = \sum_{S\in\Sc:i,j\in S} 1/|S|$.  Then we have $A = \tfrac p L (D(\Sc) - A(\Sc))$.  This follows from observing that for subset $S$, we have $A_S A_S^+ = I_S - \tfrac 1 {|S|} \ones_S$, and counting as we sum.  Taking the complete $K$-graph yields $D(\Sc) = \binom{N-1}{K-1}I$ and $A(\Sc)_{ij} = \tfrac 1 K \binom{N-2}{K-2} = \tfrac{K-1}{K(N-1)}\binom{N-1}{K-1}$ for $i\neq j$ and $A(\Sc)_{ii} = \tfrac 1 K \binom{N-1}{K-1}$.  Putting this all together gives $A = \tfrac{1}{L \binom{N}{K}} \tfrac N K \tfrac {K-1}{N-1} \binom{N-1}{K-1} (I - \tfrac 1 N \ones)) = \tfrac {N-1}{L(K-1)} (I-\tfrac 1 N \ones)$.}

\begin{figure}[!h]
\includegraphics[width=\columnwidth]{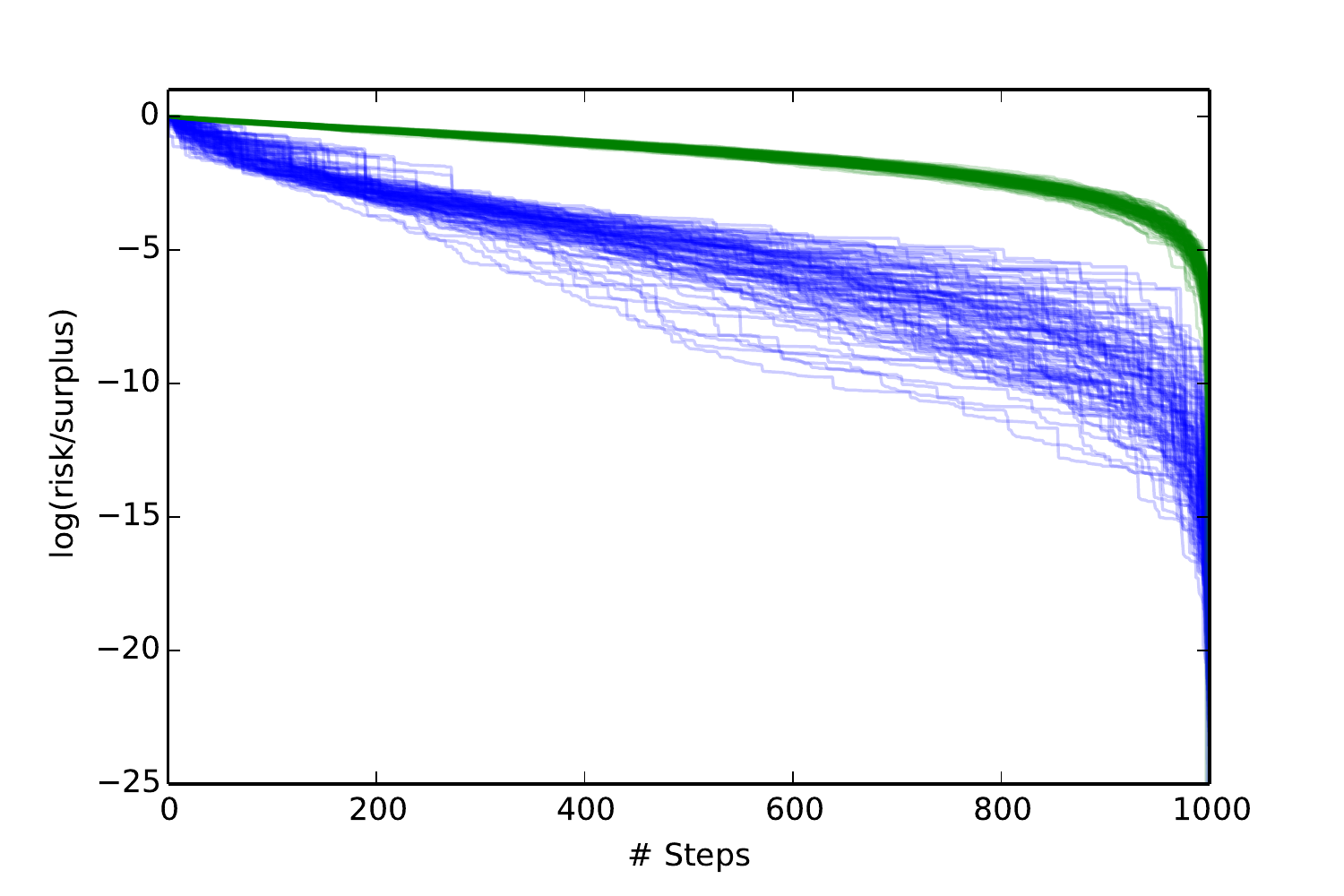}
\caption{Log of fraction of captured network surplus vs. number of trading interactions for fair edge dynamic (green lines; top) and fair node dynamic (blue lines; below). \label{fig:risk_vs_steps}}
\end{figure}

\end{document}